\documentclass[a4paper,reqno]{amsart}
\usepackage{graphicx}
\date{February 11, 2025}
\usepackage{amsthm,amssymb,bbm}

\def\rz{\mathbb{R}}

\def\gD{\mathfrak{D}}
\def\gR{\mathfrak{R}}
\def\gZ{\mathfrak{Z}}

\def\cE{\mathcal{E}}
\def\cF{\mathcal{F}}

\def\ri{\mathrm{i}}
\def\rd{\mathrm{d}}

\newtheorem{corollary}{Corollay}
\newtheorem{definition}{Definition}
\newtheorem{lemma}{Lemma}
\newtheorem{theorem}{Theorem}

\title[Excess Charge]{Bound on the Excess Charge of Generalized
  Thomas-Fermi-Weizs\"acker Functionals}

\author{Rafael D.~Benguria}
\author{Heinz Siedentop}
\begin{document}

\address{Rafael D. Benguria\\ Instituto de F\'isica\\
  Pontificia Universidad Cat\'olica de Chile\\
  Av.~Vicu\~na Mackenna 4860\\ Macul\\ 7820436, Santiago\\ Chile}
\email{rbenguri@uc.cl} \address{Heinz Siedentop\\ Mathematisches
  Institut\\ Ludwig-Maximilians-Universit\"at M\"unchen\\
  Theresienstr. 39\\ 80333 M\"unchen, Germany\\ and Munich Center for
  Quantum Science and Technology (MCQST)\\ Schellingstr. 4\\ 80799
  M\"unchen\\ Germany} \email{h.s@lmu.de}

\begin{abstract}
  We bound the number of electrons $Q$ that an atom can bind in excess
  of neutrality for density functionals generalizing the classical
  Thomas-Fermi-Weizs\"acker functional: instead of the classical power
  $5/3$ more general powers $p$ are considered. For $3/2<p<2$ we prove
  the excess charge conjecture, i.e., that $Q$ is uniformly bounded in
  the atomic number $Z$. The case $p=3/2$ is critical: the behavior
  changes from a uniform bound in $Z$ to a linear bound at the
  critical coupling $4\sqrt\pi$ of the nonlinear term. We also improve
  the linear bound for all $p\geq6/5$.
\end{abstract}

\maketitle
\tableofcontents
\section{Introduction\label{Einleitung}}
	
Density functionals have been essential tools to analyze physical
properties of atoms, molecules, and physics. In quantum mechanics on
can trace them back to Thomas \cite{Thomas1927} and Fermi
\cite{Fermi1927,Fermi1928}. Weizs\"acker \cite{Weizsacker1935} added
an inhomogeneity correction meant to improve the behavior of the
density in regions of rapid change of the external potential. Benguria
et al. \cite{Benguriaetal1981} generalized the classical
Thomas-Fermi-Weiz\-s\"acker functional to a more general functional and
initiated its mathematical analysis. Written in the square root $\psi$
of the electron density $\rho$ it reads
\begin{equation}
  \label{eq:1}
  \begin{split}
    &\cE_{p}: H^1(\rz^3:\rz)\to \rz\\
    &\psi \mapsto \underbrace{A\int_{\rz^3} |\nabla\psi|^2}_{T[\psi]:=}
      +\underbrace{\frac\gamma{p}\int_{\rz^3}|\psi|^{2p}}_{F[\psi]:=}
      -\underbrace{\int_{\rz^3}V\psi^2}_{A[\psi]:=}
           + \underbrace{\frac12\int_{\rz^3}\rd x  \int_{\rz^3} \rd y{\psi(x)^2\psi(y)^2\over|x-y|}}_{D[\psi^2]:=}.
  \end{split}
\end{equation}
with $p\geq1$, $\gamma>0$, and
\begin{equation}
  \label{eq:V}
  V(x)=\sum_{k=1}^K{Z_k\over |x-R_k|}
\end{equation}
where $\gZ:=(Z_1,...,Z_K)\in \rz_+^K$ and $\gR:=(R_1,...,R_K)\in \rz^{3K}$.

The generalized Thomas-Fermi term $F$ arises by a semiclassical
approximation from a many electron Hamiltonian with kinetic energy
operator $|-\ri \nabla|^\alpha$. The relation between the exponents is
\begin{equation}
  p={\alpha+3\over3}
\end{equation}
with the prefactor given by
\begin{equation}
  \gamma=(3\pi^2)^\frac\alpha3.
\end{equation}
The range $2>p>4/3$ corresponds to $3>\alpha>1$.

The Euler equation for the minimizer $\psi$ reads
\begin{equation}
  \label{eq:2}
  -A\Delta\psi +[\gamma|\psi|^{2p-2}-(\underbrace{V-|\psi|^2*|\cdot|^{-1}}_{\varphi:=})]\psi =0 
\end{equation}
weakly in $H^1(\rz^3:\rz_+)$.

For $p>3/2$ we can scale the Euler equation by making the ansatz
\begin{equation}
  \label{eq:scaling}
  \psi(x)= a_p\tilde\psi(b_px),\ Z_k= c_p\tilde Z_k,\ R_k= \tilde R_k/b_p
\end{equation}
with
\begin{equation}
a_p := {A^\frac1{4p-6}\over\gamma^\frac1{2p-3}},\ b_p := {A^\frac{2-p}{4p-6}\over\gamma^\frac1{4p-6}},\ c_p := {A^\frac{3p-4}{4p-6}\over \gamma^\frac1{4p-6}}.
\end{equation}
Then
\begin{equation}
  \label{eq:scaledEuler}
-\Delta\tilde\psi +[|\tilde\psi|^{2p-2}-(V_{\tilde\gR,\tilde\gZ}-|\psi|^2*|\cdot|^{-1})]\tilde\psi =0.
\end{equation}
Thus we may assume for $p>3/2$ that $A=\gamma=1$ and retrieve the
general case by \eqref{eq:scaling}.

If $\psi$ minimizes $\cE_{p}(H^1(\rz^3:\rz))$ then we write
$N:=\int_{\rz^3}\psi^2$ for its particle number $Z:=Z_1+...+Z_K$ and
\begin{equation}
  \label{eq:Q}
  Q:=N-Z
\end{equation}
for its excess charge. The excess charge has been studied
previously. Lower bounds on the excess charge are $Q\geq0$ for
$p\geq 4/3$ \cite[Lemma 12]{Benguriaetal1981} and $Q>0$ for $p>5/3$
and $K=1$ \cite[Lemma 13]{Benguriaetal1981}.

The following is known about upper bounds: In the classical case
Benguria and Lieb \cite[Formula (43)]{BenguriaLieb1985} showed
$Q\leq 270.74 K$. Solovej \cite[Proposition 14]{Solovej1990} improved
this value to $Q\leq 178.03K$. Based on an immediate generalization of an
argument of Lieb \cite{Lieb1984} it is clear that $Q<Z$ for all $p$.

The focus of this work will be on generalizing and sharpening those
upper bounds. We will start in Section \ref{s2} with the
improvement $Q\leq 0.5211 \, Z$ for rather general $p$.

We will also generalize the bound uniform in the nuclear charge $Z$ similar
to \cite[Formula (43)]{BenguriaLieb1985} and Solovej
\cite[Proposition 14]{Solovej1990}. We will carry these results
through for $p\in(3/2,2)$ in Section \ref{s3}.

The value $p=3/2$ is critical, since the energetic dominance shifts
from the Thomas-Fermi term for $p>3/2$ to the Weizs\"acker term for
$p<3/2$. We will treat the critical case by different methods and will
show for $\gamma<\gamma_c:=4\sqrt\pi$ a bound proportional to $Z$
whereas for $\gamma\geq\gamma_c$ we have $Q=0$. This is done in
Section \ref{s4}.

\section{Improving the bound $Q\leq Z$ for exponents
  $p\geq6/5$}\label{s2}

We consider the atomic case, i.e., $K=1$. Because of translational
invariance we can and will assume $R_1=0$ throughout this
section. Following Benguria and Tubino \cite{BenguriaTubino2022} we
will improve \eqref{eq:s4}. We begin with an inequality by Nam
\cite{Nam2012}
\begin{equation}
  \label{eq:s21}
  \beta :=\inf\left\{\frac{\frac12 \int_{\mathbb{R}^3} \frac{|x|^2+|y|^2}{|x-y|} \, \psi(x)^2 \psi(y)^2 \, \rd x \rd y }{\int_{\mathbb{R}^3}\psi(x)^2 \, \rd x  \int_{\mathbb{R}^3}|x| \, \psi(x)^2\, \rd x} \Big| 0\neq\psi\in H^1(\rz^3:\rz)\right\}
  \ge 0.8218.
\end{equation}
\begin{theorem} 
  Let $\psi$ be a non-vanishing solution of \eqref{eq:2} for $K=1$
  in $H^1(\rz^3)$. Then, for all $\gamma \ge 0$ and all $p \ge 6/5$,
\begin{equation} 
\int_{\mathbb{R}^3} \psi(x)^2 \, \rd x \le \frac{5}{4 \beta} Z \le 1.5211 \, Z.
\label{eq:s5}
\end{equation}
\end{theorem}

\begin{proof}
  Since the groundstate energy of hydrogen is $-1/4$ (Schr\"odinger
  \cite[Equation (19)]{Schrodinger1926I}), we have for any positive
  nuclear charge $\tilde Z$
\begin{eqnarray}
  \label{eq:RR}
  \int_{\rz^3}|\nabla f|^2 - \int_{\rz^3} {\tilde Z\over|x|}|f(x)|^2\rd x \geq -{\tilde Z^2\over4}\int_{\rz^3}|f|^2
\end{eqnarray}
for any $f\in H^1(\rz^3)$.  Picking $f:=\psi$ this can be recast as
\begin{equation}
  \frac1Z A[\psi]\leq \frac1{\tilde Z} T[\psi]+ {\tilde Z \over 4}N
\end{equation}
using the notation for the various parts of the energy in \eqref{eq:1}.
Optimizing in $\tilde Z$ yields
\begin{equation}
\left(\frac{A}{Z}\right)^2 \le K \, N.
\label{eq:s13}
\end{equation}
Now, from (\ref{eq:s11}) and (\ref{eq:s13}) we get
\begin{equation}
A\le \frac13 N \, Z^2.
\label{eq:s14}
\end{equation}
If we define 
\begin{equation}
I=\frac{ \int_{\mathbb{R}^3}|x| \, \psi(x)^2\, \rd x}{ \int_{\mathbb{R}^3}\psi(x)^2 \, \rd x}=\frac{1}{N} \,  \int_{\mathbb{R}^3}|x| \, \psi(x)^2\, \rd x
\label{eq:s15}
\end{equation}
and use the Schwarz inequality, we get 
\begin{equation}
N^2 =\left(\int_{\mathbb{R}^3}\psi(x)^2 \, \rd x \right)^2 \le  \int_{\mathbb{R}^3}|x| \, \psi(x)^2\, \rd x\int_{\mathbb{R}^3} \frac{\psi(x)^2}{|x|} \, \rd x
\le \left(I \, N \right)  \frac{A}{Z}.
\label{eq:s16}
\end{equation}
Using (\ref{eq:s14}) and (\ref{eq:s16}) we finally get 
\begin{equation}
I \ge \frac{3}{Z}.
\label{eq:s17}
\end{equation}
To conclude we use Nam's method \cite{Nam2012}.  We multiply
(\ref{eq:2}) by $\psi \cdot |x|^2$ and integrate over $\mathbb{R}^3$.
From Nam's result we have
\begin{equation}
(-\Delta \psi, |x|^2 \psi) \ge -\frac34 \,(\psi,\psi) = -\frac34 N.
\label{eq:s18}
\end{equation}
Also,
$\gamma\int_{\mathbb{R}^3}|\psi(x)|^{2p-2}\psi(x)^2|x|^2\, \rd
x\ge0$. Hence,
\begin{equation}
\int_{\mathbb{R}^3} \varphi(x) |x|^2 \psi(x)^2 \, \rd x \ge -\frac34  N, 
\label{eq:s19}
\end{equation}
and using 
$$
\varphi (x) = \frac{Z}{|x|} - \int_{ \mathbb{R}^3} \frac{\psi(y)^2}{|x-y|} \, \rd y,
$$
we get
\begin{equation}
Z \int_{\mathbb{R}^3}|x| \, \psi(x)^2\, \rd x- \frac12 \int_{\mathbb{R}^3} \int_{\mathbb{R}^3} \frac{|x|^2+|y|^2}{|x-y|} \, \psi(x)^2 \psi(y)^2 \, \rd x \rd y \ge -\frac34 N,
\label{eq:s20}
\end{equation}
where we did the standard symmetrization in the second
integral. Moreover, from (\ref{eq:s15}), the first term in
(\ref{eq:s20}) is given by $ZNI$.

From the definition of $\beta$, $N$, and $I$, we have
\begin{equation}
\frac12 \int_{\mathbb{R}^3} \frac{|x|^2+|y|^2}{|x-y|} \, \psi(x)^2 \psi(y)^2 \, \rd x \rd y \ge \beta N^{2} I.
\label{eq:s22}
\end{equation}
From (\ref{eq:s20}),  (\ref{eq:s22}), and the definition of $I$, we get
\begin{equation}
Z  N  I \ge \beta N^2 I-\frac34 N \ge  \beta N^2 I -\frac14  N Z I
\label{eq:s23}
\end{equation}
where the last inequality follows from (\ref{eq:s17}). Rearranging
(\ref{eq:s23}), we get
$$
\frac54 Z  N  I \ge \beta N^2 I
$$
yielding
\begin{equation}
N \le \frac{5}{4 \beta} Z.
\label{eq:s24}
\end{equation}
Using Nam's numerical lower bound \eqref{eq:s21} on $\beta$ in
(\ref{eq:s24}) gives
\begin{equation}
N\le 1.5211 Z.
\label{eq:s25}
\end{equation}
\end{proof}

\section{Some preparatory results \label{newsection}}

We start with a central observation keeping the constant $\gamma$ in
this section, since we will use the result also for $p=3/2$ where it
cannot be scaled out. Set
\begin{definition}
  \begin{equation}
    \label{eq:p}
    \begin{split}
      P:\rz^3\setminus\left\{R_1,...,R_K\right\}&\to\rz_+\\
        x&\mapsto \sqrt{4\pi\psi(x)^2+\varphi(x)^2}.
    \end{split}
  \end{equation}
\end{definition}
The following allows us to bound the excess charge from above.
\begin{lemma}
  The function $P$ is subharmonic, and, for $K=1$ and $R_1=0$, the function
  $rP(r))$ is convex, monotone decreasing in $r$, and
  $\lim_{r\to\infty}rP(r)=Q$. In particular, for all $r>0$
  \begin{equation}
    \label{eq:rp}
    rP(r) \geq Q.
  \end{equation}
\end{lemma}
Note that we write -- in abuse of notation -- $P(r)$ instead of $P(x)$
in the radial case.
\begin{proof}
  We compute
  \begin{equation}
    \label{Deltap2}
    \begin{split}
      &2P\Delta P+ 2(\nabla P)^2=\Delta P^2= 8\pi\psi\Delta\psi 
      +8\pi(\nabla\psi)^2 + 2\varphi\Delta\varphi + 2 (\nabla\varphi)^2\\
      =&8\pi\psi^2(\gamma\psi^{2p-2}-\varphi)+8\pi(\nabla\psi)^2+8\pi\varphi\psi^2 + 2 (\nabla\varphi)^2\\
      =&8\pi\psi^2\gamma\psi^{2p-2}+8\pi(\nabla\psi)^2+ 2 (\nabla\varphi)^2
     \end{split}
   \end{equation}
   where we used the TFW equation \eqref{eq:2}. By Schwarz'
   inequality $(\nabla P)^2 \leq 4\pi(\nabla\psi)^2+(\nabla\varphi)^2$
   and therefore \eqref{Deltap2} implies $\Delta P\geq0$. This shows
   the subharmonicity.

  Since
  \begin{equation}
    \Delta P(x) = \frac1r{\rd^2\over \rd r^2} r P(r)
  \end{equation}
  the convexity and monotonicity statements follow.

  By Lemma \ref{abfall} $r\psi(r)\to0$ and by Newton's lemma
  $r\varphi(r)\to-Q$ as $r\to\infty$. This shows the limiting
  statement. -- Inequality \eqref{eq:rp} is an immediate consequence
  of the monotonicity and the limiting statement.
\end{proof}

\begin{lemma}
  For $\lambda \in (0,1)$ and $3/2<p<2$, every minimizer
  $\psi$ of $\cE_p$ with $\varphi$ and $V$ given by
  (\ref{eq:1}) and (\ref{eq:V}) respectively, satisfies,
  \begin{equation}
    \lambda \gamma \psi^{2p-2} \le \varphi + c_p(\lambda)  \gamma^{1/(3-2p)},
    \label{eq:t1}
  \end{equation}
  with
  \begin{equation}
    \label{eq:t2}
    c_p(\lambda) = {\left(2\pi\right)^{\frac{p-1}{2p-3}} \over \lambda^{\frac{p-1}{2p-3}} (1-\lambda)^{\frac{2-p}{2p-3}}} (2p-3) {(2-p)^{\frac{2-p}{2p-3}}\over (p-1)^\frac{2p-2}{2p-3}}.
\end{equation}

\end{lemma}

Note that when $p=5/3$, i.e., in the case of non-relativistic TFW theory,
\begin{equation}
c_\frac53(\lambda)=\frac94 \, \pi^2 \frac{1}{\lambda^2(1-\lambda)}
\end{equation}
as in \cite[Formula (13)]{BenguriaLieb1985}.

\begin{proof} 
Put $u:=\psi^{2p-2}$. Then (\ref{eq:2}) implies that
\begin{equation}
-\Delta u +  (2p-2) (\gamma u -\varphi) u \le 0, 
\label{eq:t4}
\end{equation}
provided $p \ge 3/2$. On the other hand (\ref{eq:V}) and (\ref{eq:2}) imply 
\begin{equation}
-\Delta \varphi = -4\pi \psi^2 = - 4\pi u^{1/(p-1)}, 
\label{eq:t5}
\end{equation}
away from $R_k$, $k=1,2, \dots K$. Now, set 
\begin{equation}
v := \gamma \lambda u -\varphi -d,
\label{eq:t6}
\end{equation}
with $d$ a constant to be chosen later. Then, from (\ref{eq:t5}) and
(\ref{eq:t6}) we get
\begin{equation}
  -\Delta v= -\gamma \lambda\Delta u + \Delta \varphi \le -\gamma \lambda (2p-2)(\gamma u -\varphi) u+ 4 \pi u^{\frac{1}{p-1}}.
\label{eq:t7}
\end{equation}
Set $S=\{x |v > 0\}$. It follows from (\ref{eq:t6}) and (\ref{eq:V})
that $R_k \notin S$, all $k=1, 2, \dots ,K $. On $S$,
\begin{equation}
 \varphi =\gamma \lambda  \, u - d - v \le -d +\gamma \lambda  \, u, 
\label{eq:t8}
\end{equation}
and replacing in (\ref{eq:t7}), on $S$, 
$$
-\Delta v \le -\gamma \lambda(2p-2)\gamma(1-\lambda)u^2 -\gamma
\lambda(2p-2) d \, u + 4 \pi u^{\frac{1}{p-1}}
$$
Now, if $3/2<p<2$, then $1/2<p-1 <1$,  $1<1/(p-1)<2$, and 
$$
4 \pi u^{\frac{1}{p-1}} \le {\gamma}^2 u^2 \lambda(1-\lambda)(2p-2) + b \, u, 
$$
i.e.,
$$
4 \pi u^{\frac{2-p}{p-1}} \le {\gamma}^2 u \lambda(1-\lambda)(2p-2) + b. 
$$
Consider the function 
\begin{equation}
f(u) =4 \pi u^{\frac{2-p}{p-1}}  - {\gamma}^2 u \lambda(1-\lambda)(2p-2),
\label{eq:t20}
\end{equation}
with $(2-p)/(p-1) <1$, i.e., $p>3/2$. Since $(2-p)/(p-1) <1$, the function $f(u)$ for $u>0$ has only one maximum, say $\hat u$, on the interval $(0,\infty)$. 
$\hat u$ is given by, 
$$
\hat u = \left[\frac{{\gamma}^2 (p-1)^2  \lambda(1-\lambda)}{2\pi(2-p)}\right]^{(p-1)/(3-2p)},
$$
and, 
$$
f(\hat u) = \left(2 \pi\right)^{\frac{p-1}{2p -3}}\left({\gamma}^2  \lambda(1-\lambda)\right)^{\frac{2-p}{3-2p}}2(2p-3)(p-1)^{\frac{1}{3-2p}}(2-p)^{\frac{2-p}{2p-3}}.
$$
Now choose $b=f(\hat u)$, which in turn implies choosing $d=b/(\gamma \lambda(2p-2))$.
With that choice of $b$, $-\Delta v \le 0$ on $S$. Hence $v$ is subharmonic on $S$. Since $v=0$ on $\partial S$, we conclude that $v<0$ on $S$.
This in turn implies that $S$ is empty and we are done.
\end{proof}

\begin{corollary}
  Assume $x\in\rz^3$ such that $\varphi(x)\leq 0$. Then
  \begin{equation}
    \label{eq:fi-negative}
    \psi(x)\leq {2^{4p-3\over 4p-6}\pi^{p-1\over2p-3}(2p-3)^{1\over2p-2}\over(3p-4)^{3p-4\over2(p-1)(2p-3)}}. 
  \end{equation}
\end{corollary}
\begin{proof}
  Drop $\varphi$ in \eqref{eq:t1}, move all $\lambda$ to the
  right, and minimize the right side in $\lambda$. The minimum occurs
  at $(3p-4)/(2p-2)$. The minimal value is the right hand side of
  \eqref{eq:fi-negative}.
\end{proof}

We also have the bound
\begin{lemma}
   \label{lemmag1}
   For $p\geq3/2$ and positive $\gamma$ we have
   \begin{equation}
     \label{eq:psi-V}
     \psi^{2p-2}\leq \gamma^{-1}V.
   \end{equation}
 \end{lemma}
 \begin{proof}
   First we note that $V$ is harmonic outside the positions of the nuclei $R_1,...,R_K$. We set
   \begin{equation}
     \label{eq:f}
     f:=\gamma^{-1}V-\psi^{2p-2}.
   \end{equation}
   We wish to show that $f\geq0$ and therefore consider the exceptional set
   \begin{equation}
     \label{eq:S}
     S:=\left\{x\in\mathbb{R}^3|f(x)<0\right\}.
   \end{equation}
   Certainly neither of the $R_1,...,R_K$ is in $S$, since $\psi$ is finite everywhere. On $S$ we get from \eqref{eq:2}
   \begin{multline}
     \Delta f=-\Delta\psi^{2p-2}=-(2p-2)(2p-3)\psi^{2p-4}(\nabla\psi)^2-(2p-2)\psi^{2p-3}\Delta\psi\\
     \leq (2p-2)\psi^{2p-3}(\varphi-\gamma\psi^{2p-2})\psi\leq(2p-2)\psi^{2p-2} (V-\gamma\psi^{2p-2})<0      
   \end{multline}
   on $S$.  Hence $\Delta f<0$ on $S$, i.e., $f$ is superharmonic on
   $S$. Since $f$ vanishes on the boundary of $S$, the exceptional set
   $S$ is empty which implies the claim.
 \end{proof}

 For $R>0$ define $D_R$ as the complement of the union of the
balls of radius $R$ around the positions of the nuclei $R_1,...,R_K$,
i.e., $B_R(R_1)\cup...\cup B_R(R_K)$.

\begin{lemma}
  For $3/2<p<2$, $\gamma>0$ and $R>0$, let $\psi$ be the
  positive solution of (\ref{eq:2}) with $V$ given by (\ref{eq:V}),
  $\varphi$ by \eqref{eq:2}, and $S_{p,R}$ by
  \eqref{sommerfeldR}. Then for all $x\in D_R$,
\begin{equation}
\varphi(x) \le \frac{\pi^2}{R^2} + \sum_{k=1}^Ks_{p,R}(|x-R_k|).
\label{eq:z0}
\end{equation}

\end{lemma}

\begin{proof}
  Let $W=\gamma \, \rho^{p-1} - \varphi$, $\rho=\psi^2$ and consider the
  Hamiltonian $H = -\Delta+ W$. The operator $H$ is non-negative, since
  its ground state, the function $\psi$, has zero energy. In particular
  \begin{equation}
    \int_{\mathbb{R}^3} \left(|\nabla e_{R,a}|^2+ W |e_{R,a}|^2\right)\geq0
\label{eq:z1}
\end{equation}
where $e_{R,a}$ is the positive normalized groundstate of the
Dirichlet Laplacian on $B_R(a)$, $a\in\rz^3$, extended by $0$ to the
complement of $B_R(a)$, i.e.,
\begin{equation}
e_{R,a}(x) =\begin{cases}
   {1\over R^\frac32 \sqrt{2 \pi }}{\sin\left(\pi |x-a|/R\right)\over|x-a|/R}& |x-a|<R\\
  0&|x-a|\geq R.
\end{cases}
\end{equation}
Obviously $e_{R,a}\in H^1(\rz^3)$, is spherically symmetric about $a$,
is decreasing, and has compact support.  With $g_{R,a}:=e_{R,a}^2$ and
$g_R:=g_{R,0}$ we have
\begin{equation}
  \int_{\mathbb{R}^3} |\nabla e_{R,x}|^2=
\left(\frac{\pi}{R}\right)^2.
\end{equation}
Thus, (\ref{eq:z1}) implies for all $R>0$ and all $x\in\rz^3$
\begin{equation}
\int_{\mathbb{R}^3} W(y) e_{R,x}(y)\, \rd y \ge -\left(\frac\pi R\right)^2.
\label{eq:z2}
\end{equation}
Note that
$\int_{\mathbb{R}^3} W(y) g_{R,x}(y)\, \rd y=\left(g_R \ast W \right)
(x)$ where $\ast$ denotes convolution. Define
\begin{equation}
  \tilde \phi =\varphi \ast g_R - \left(\frac\pi R\right)^2.
  \label{eq:z3}
\end{equation}
Since $\varphi \in L^{3+\epsilon} + L^{3-\epsilon}$, $\epsilon>0$
(\cite{Benguriaetal1981}, proof of Lemma 7) and $g_R \in L^s$ for all
$s \ge 1$, $\tilde \phi$ is continuous and tends to $-(\pi/R)^2$ at
infinity (see \cite{Lieb1981}, Lemma 3.1). Using H\"older's
inequality, we have for all $x$
\begin{equation} 
\left(g_R\ast\rho^{p-1}\right)(x) \le \left[\left( g_R\ast \rho\right)(x)\right]^{p-1} \, \left(\int g_R(y) \, \rd y \right)^{2-p}= \left[\left( g_R\ast \rho\right)(x)\right]^{p-1}.
\label{eq:z4}
\end{equation}
provided $1\le p \le 2$. Here we used $\int_{\rz^3} g_R=1$. Let us
also define
\begin{equation}
\tilde \rho = g_R\ast \rho.
\label{eq:z5}
\end{equation}
From equations (\ref{eq:z2})--(\ref{eq:z5}) we obtain for all $x$, 
\begin{equation}
  \left(\frac\pi R\right)^2 \ge \left(\varphi \ast g_R\right)(x)-\gamma \left(g_R\ast\rho^{p-1}\right)(x) \ge \tilde \phi(x) +\left(\frac\pi R\right)^2 - \gamma{\tilde \rho}(x)^{p-1}.
\label{eq:z6}
\end{equation}
In other words 
\begin{equation}
\tilde \phi \le \gamma {\tilde \rho \,}^{p-1}
\label{eq:z7}
\end{equation}
provided $1\le p \le 2$.  Notice that $\varphi$ is subharmonic away from
the nuclei and that $\tilde \varphi= g_R\ast \phi -(\pi/R)^2$
with $g_R$ being
spherically symmetric, positive, of total mass one, and having support
in a ball of radius $R$. From this it follows that
\begin{equation}
\varphi(x) \le \tilde \phi(x) +\left(\frac\pi R\right)^2, 
\label{eq:z8}
\end{equation}
for all $x$ such that $|x-R_k|>R$, for all $k$. To prove (\ref{eq:z0})
we need a bound on $\tilde \phi$. From (\ref{eq:V}) and (\ref{eq:z3}),
using the bound (\ref{eq:z7}), and the fact that the Laplacian
commutes with convolution, we compute
\begin{equation}
-\frac{1}{4 \pi} \Delta \tilde \phi = \tilde V - \tilde \rho \le  \tilde V - {\gamma}^{-1/(p-1)} \left[{\tilde\phi}_{+}(x)\right]^{1/(p-1)}
\label{eq:z9}
\end{equation}
with 
\begin{equation}
\tilde V =V*g_R
\label{eq:z10}
\end{equation}
and with
${\tilde\phi}_{+}(x) = {\rm max}\left(\tilde \phi(x),0 \right)$. Let
$\hat \phi$ be the minimizer of the Thomas-Fermi functional with
external potential $\tilde V$. It fuffills the equation
\begin{equation}
  -\frac{1}{4 \pi} \Delta \hat \phi =  \tilde V - \left[{\hat \phi}_{+}(x)\over \gamma\right]^{1/(p-1)}.
\label{eq:z11}
\end{equation}
By the maximum principle we have for all $x$
\begin{equation}
\tilde \phi(x) \le \hat \phi(x).
\label{eq:z12}
\end{equation}

The next step is to bound $\hat \phi$. We treat first the radial
case with $V(x)=Z/|x|$. Since the Sommerfeld solution $S_{p,R}$ of the generalized TF model
defined in \eqref{sommerfeldR} fulfills for $p\in(3/2,2)$
\begin{equation}
\frac{1}{4\pi} \Delta S_{p,R} \le {\gamma}^{-1/(p-1)}\, S_{p,R}^{1/(p-1)}
\label{eq:z16}
\end{equation}
for $r>R$ and $\hat \phi$ satisfies there
\begin{equation}
\frac{1}{4\pi} \Delta {\hat \phi} = {\gamma}^{-1/(p-1)}\, {\hat \phi}^{1/(p-1)}.
\label{eq:z17}
\end{equation}
we can again use a comparison argument. Since $s_{p,R}(R) - {\hat \phi}(R)=\infty$ we conclude that
\begin{equation}
{\hat \phi} (r) \le s_{p,R}(r) \qquad \mbox{for $r>R$}.
\label{eq:z18}
\end{equation}
This, together with (\ref{eq:z8}) and (\ref{eq:z12}) proves
(\ref{eq:z0}) in the radial case.

For the non radial case, let ${\hat \phi}_j(x)$ be the solution to
(\ref{eq:z11}) for an atom of smeared nuclear charge at $R_j$. By
another comparison argument (see Lieb and Simon \cite[Theorem
V.12]{LiebSimon1977} or Lieb \cite[Corollary 3.6]{Lieb1981}) we get
$$
{\hat \phi} (x) \le \sum_{j=1}^K {\hat \phi}_j(x).
$$
This, together with the definition of $S_{p,R}$ in \eqref{sommerfeldR} and
(\ref{eq:z8}) implies (\ref{eq:z0}).
\end{proof}

 \section{Bound on the excess charge uniform in the atomic number for
   exponents $p\in(3/2,2)$ \label{s3}}
 Next we turn to bounds on $Q$ that are uniform in the atomic
 number. That such bounds exist is the content of the excess charge
 conjecture. We will prove it for $p\in(3/2,2)$ following
 \cite{BenguriaLieb1985} and start with the atomic case where we get
 \begin{theorem}
   Assume $p\in(3/2,2)$ and $B(p)$ as given in \eqref{eq:B}. Then the
   atomic excess charge $Q$ of the generalized TFW functional is
   bounded as follows:
   \begin{equation}
     0\leq Q\leq B(p) {A^\frac{3p-4}{4p-6}\over \gamma^\frac1{4p-6}}.
   \end{equation}
 \end{theorem}
 \begin{proof}
We recall that $\varphi= V- \rho*|\cdot|^{-1}$ with $\rho$ the minimizer
of $\cE_p$ is the electric mean-field potential of the generalized TFW
minimizer (see \eqref{eq:2}). Similarly we write
\begin{equation}
  \label{TFpot}
  \hat\phi=\tilde V-\rho^\mathrm{TF}*|\cdot|^{-1}
\end{equation}
where $\rho^\mathrm{TF}$ is the minimizer of $\cE_p^\mathrm{TF}$ with
external potential $\tilde V$.

We have the following simple bound in terms of $s_{p,R}$
\begin{equation}
  \begin{split}
  &Q\leq r \sqrt{4\pi \psi^2+\varphi^2}
  \leq r \sqrt{4\pi \left({\varphi+c_p(\lambda)\over\lambda}\right)^\frac1{p-1}+\varphi^2}\\
  \leq &r\sqrt{4\pi\left({\hat\phi+{\pi^2\over R^2}+c_p(\lambda)\over\lambda} \right)^{\frac1{p-1}}+\left(\hat\phi+{\pi^2\over R^2} \right)^2}\\
  \leq& 
  \begin{cases} F(p,\lambda,R,r):=r\sqrt{4\pi\left({s_{p,R}+{\pi^2\over
R^2}+c_p(\lambda)\over\lambda}
\right)^{\frac1{p-1}}+\left(s_{p,R}+{\pi^2\over R^2} \right)^2} & \varphi(r)\geq0\\
G(p,\lambda,R,r):= r\sqrt{4\pi\left({c_p(\lambda)\over\lambda}\right)^{\frac1{p-1}}
  +\left(s_{p,R}+{\pi^2\over R^2} \right)^2}& \varphi(r)<0.
\end{cases}
  \end{split}
\end{equation}
(Note that we could replace using Lemma \ref{simple-bound} by Lemma
\ref{better-bound} which would improve the numerical result. However,
to keep the numerical evaluation simple we refrain from doing so.)
Thus
\begin{equation}
  \label{eq:B}
  Q\leq B(p):= \max\left\{\min_{0<\lambda<1,\ 0<R<r}F(p,\lambda,R,r),\ \min_{0<\lambda<1,\ 0<R<r} G(p,\lambda,R,r)\right\}.
\end{equation}
The parameters $A$ and $\gamma$ are restored by the scaling relations \eqref{eq:scaling}.

\begin{figure}[h]
 \centering % centering figure
 \scalebox{0.5} % rescale the figure by a factor of 0.8
 {\includegraphics{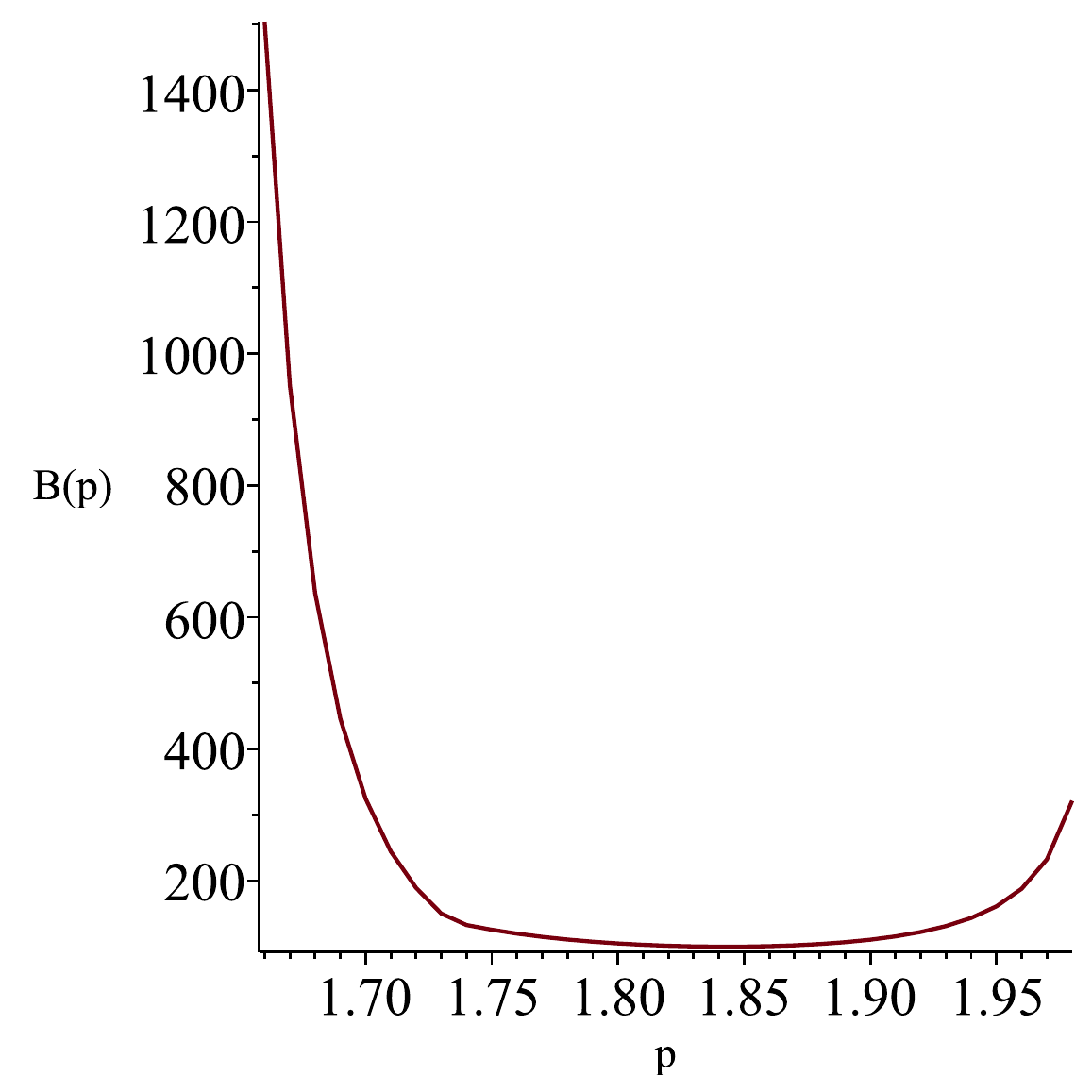}} % importing figure
 \caption{{\small Upper bound $B(p)$ on the excess charge. Minimum at $p_\mathrm{min} \approx  1.8431$ with $B(p_\mathrm{min})\approx 101.14$.}}\label{fig:exm}
\end{figure}
\end{proof}

The general case is more or less a corollary of the atomic
case (see \cite{BenguriaLieb1985}) and merely yields a factor $K$, i.e.,
\begin{equation}
  \label{eq:mol}
  Q\leq B(p) K.
\end{equation}

\section{The critical exponent $p=3/2$} \label{s4}

We turn to the critical exponent $p=3/2$. We will keep the parameter
$\gamma$ in $\cE_p$ in this section, since it cannot be scaled out in
this case. In fact our result will depend on $\gamma$. Lemma
\ref{lemmag1} allows us to prove the following bound.
\begin{theorem}
  \label{theorem2}
  Let $\psi$ be the positive solution of \eqref{eq:2} for $p=3/2$. Then
  \begin{equation}
    \label{eq:haupt2}
    Q\leq
    \begin{cases}
      0 &\gamma\geq\gamma_c:=4\sqrt\pi\\
      {\gamma_c-\gamma\over\gamma}Z& \gamma<\gamma_c.
    \end{cases}
  \end{equation}
\end{theorem}
\begin{proof}
  We set
  \begin{equation}
    \label{eq:8a}
    g:=\varphi+aV-b\psi
  \end{equation}
  with non-negative constants $a$ and $b$ to be suitably specified
  later. We want to show that $g$ is nonnegative on all of
  $\mathbb{R}^3$. Again, we use a subharmonic argument and define an
  exceptional set
  \begin{equation}
    \label{eq:S1}
    S:=\left\{x\in\mathbb{R}^3|g(x)<0\right\}.
  \end{equation}
  As in the proof of Lemma \ref{lemmag1}, the positions of the nuclei $R_1,...,R_K\notin S$. On $S$ we have
  \begin{equation}
    \label{eq:9}
    \Delta g =4\pi\psi^2-b\Delta\psi = 4\pi\psi^2+b\varphi\psi-b\gamma\psi^2.
  \end{equation}
Because of Lemma \ref{lemmag1} we have on $S$
\begin{equation}
  \varphi<b\psi-aV\leq b\psi-a\gamma\psi.
\end{equation}
Combining this with \eqref{eq:9} yields
\begin{equation}
  \begin{split}
  &\Delta g \leq 4\pi\psi^2+b^2\psi^2-ba\gamma\psi^2-b\gamma\psi^2\\
  =&\psi^2\left[4\pi+\left(b-\frac\gamma2(1+a)\right)^2-{\gamma^2\over4}(1+a)^2\right]\\
  = &\psi^2(4\pi-{\gamma^2\over4}(1+a)^2) ={\psi^2\over4}(\gamma_c^2-{\gamma^2\over4}(1+a)^2)
  \end{split}
\end{equation}
where we completed the square in $b$ and picked $b=\gamma(1+a)/2$.

We will make different choices of $a$ depending whether $\gamma\geq\gamma_c$
or $\gamma<\gamma_c$:

1. $\gamma\geq\gamma_c$: We chose $a=0$ which implies that $\Delta g\leq0$, i.e, $g$ is superharmonic, on $S$. Since $g$ vanishes $\partial S$, $g\geq0$ on $S$, eventually implying $S=\emptyset$. Thus,
\begin{equation}
  \varphi\geq b\psi
\end{equation}
and therefore
\begin{equation}
  \label{eq:limit}
  Q=N-Z = -\lim_{x\to\infty}|x|\varphi(x)\leq0
\end{equation}
yielding the first claim.

2. $0\leq\gamma<\gamma_c$ we pick $a:=(\gamma_c-\gamma)/\gamma>0$ and carry out the same subharmonic argument as before yielding
\begin{equation}
  \label{eq:14a}
  \varphi\geq b\psi-{\gamma_c-\gamma\over\gamma}V.
\end{equation}
 Since $\psi(x)|x|\to0$ as $x\to\infty$ by Lemma \ref{abfall}, we have
\begin{equation}
  \label{eq:fi2}
  Q=-\lim_{x\to\infty}|x|\varphi(x) \leq {\gamma_c-\gamma\over\gamma}Z
\end{equation}
which proves the claimed inequality for $\gamma<\gamma_c$.
\end{proof}

{\sc Acknowledgements:} This work has been supported by the Deutsche
Forschungsgemeinschaft (DFG), grant EXC-2111-390814868 and DFG TRR
352, grant 470903074, and FONDECYT (Chile) Project \# 124--1863.  RB
thanks the Mathematics Institute of the Ludwig--Maximilians University
and HS thanks the Physics Institute of the Pontificia Universidad
Cat\'olica de Chile for their hospitality.

\appendix
\section{Minimizers of generalized TF(W) functionals and associated
  Sommerfeld Formulae}
We collect a few known facts for $\cE_p$ mostly
from Benguria et al \cite{Benguriaetal1981}) and generalize the
Sommerfeld formula with a remainder term from $p=5/3$ (see Solovej
\cite[Lemma 4.4]{Solovej2003}) to $p\in(4/3,2)$.

\subsection {No sign change of minimizers}
Since $T[\psi]\geq T[|\psi|]$ with equality, if and only if $\psi$ does
not change the sign, and since all other terms of $\cE_p$ are
unchanged under the substitution $\psi\to|\psi|$ any minimizer of the
functional is either non-negative or non-positive.
\subsection{Existence of TFW minimizers  and strict positivity or negativity}
By standard compactness methods one shows that $\cE$ has a minimizer
$\psi\in H^1(\rz^3:\rz)$.   By unique continuation this implies that
any minimizer of the functional is strictly positive or strictly
negative.  Conversely any positive solution of \eqref{eq:2} minimizes
$\cE$.
\subsection{Uniqueness of the minimizer and spherical symmetry for atoms}
\label{uniqueness}
Since $\cE_p(\psi)=\cE_p(-\psi)$ it is enough to
study on nonnegative functions. However, the functional restricted to
nonnegative functions can be written in terms of the density $\rho$
with $\sqrt\rho:=\psi$ which makes it strictly convex. This implies
uniqueness of the minimizer $\psi$ among the non-negative functions
(and also among the negative functions). Thus $\cE_p$ has
exactly two minimizers (and \eqref{eq:2} has exactly two non-vanishing
solutions) in $H^1(\rz^3:\rz)$, one positive and one negative. --
Moreover in the atomic case, i.e., $K=1$ and $R_1=0$, uniqueness
implies that the minimizers are spherically symmetric.

\subsection{Decay of the minimizers}
\begin{lemma}
  \label{abfall}
  Suppose $p\in[1,3]$ and $\psi\in H^1(\rz^3:\rz)$ is a nonnegative
  minimizer of $\cE_p$ for the atomic case. Then
  $\psi(x)|x|\to0$ as $x\to\infty$.
\end{lemma}
\begin{proof}
  Lieb \cite[Theorem 2.12]{Lieb1981} shows that the atomic
  Thomas-Fermi energy decreases under spherically symmetric
  rearrangement. However, this is also true of $\rho^\frac53$ is
  replaced by $\rho^p$, since all $L^p$-norms are invariant under
  spherical symmetric rearrangements and it also holds if
  $\int |\nabla\sqrt\rho|^2$ is added. Thus, by uniqueness $\psi$ is
  spherically symmetric and decreasing. -- In abuse of notation, we
  will write also $\psi(r)$ instead of $\psi(x)$ with $r:=|x|$ in the
  remainder of this proof.

  Now, suppose the claim would not be true. Then there exists a
  sequence $r_1,r_2,...$ such $r_n\to\infty$ as $n\to\infty$ and
  $B:=\lim_{n\to\infty}\psi(r_n)r_n>0$. Thus there  exists $n_0$
  such that for all $n\geq n_0$ we have $\psi(r_n)r_n\geq B/2$. Thus
  we have
  \begin{equation}
    \begin{split}
      &\infty>\int_{\rz^3}|\psi|^2 \geq \sum_{n=n_0}^\infty
      {4\pi\over3} \left(r_n^3-r_{n-1}^3\right)\psi(r_n)^2 \geq \pi
      B^2\sum_{n=n_0}^\infty
      {r_n^3-r_{n-1}^3\over r_n^2}\\
      \geq &\pi
      B^2\sum_{n=n_0}^\infty(r_n-r_{n-1}){r_n^2+r_nr_{n-1}+r_{n-1}^2\over
        r_n^2} =\infty
    \end{split}
  \end{equation}
  where the first inequality uses that $\psi$ is monotone decreasing.
  However, this is a contradiction.
  \end{proof}

\subsection{The excess charge}
We know that $Q\in(-Z,\infty)$ for all $p\geq1$. If $p\geq 4/3$, then
$Q\in[Z,\infty)$, and if $p\geq 5/3$, i.e., the power of the classical
Thomas-Fermi-Weizs\"acker functional, then $Q>0$ (Benguria et
al. \cite[Theorem 1]{Benguriaetal1981}.

Moreover,  for all $\gamma \ge 0$ and $p>1$,
\begin{equation}
  Q < Z,
  \label{eq:s4}
\end{equation}
(see Lieb \cite{Lieb1981},  Theorem 7.23).

\subsection{Virial theorems}

There are two virial theorems that relate $K$, $F$, $A$, and
$R$. Assume $\psi$ to be a minimizer of $\cE_p$. Then
\begin{equation}
T[\psi]+p\, F[\psi]- A[\psi]+2 R[\psi]=0,
\label{eq:s8}
\end{equation}
and
\begin{equation}
T[\psi]+3 F[\psi]-2 A[\psi]+5 R[\psi]=0.
\label{eq:s9}
\end{equation}
Multiplying (\ref{eq:s8}) by $5$, (\ref{eq:s9}) by $2$, and
subtracting the results we obtain,
\begin{equation}
0=3T[\psi]+ (5p-6)  F[\psi]- A[\psi].
\label{eq:s10}
\end{equation}
and, if $p\geq 6/5$,
\begin{equation}
  \label{eq:s11}
  3T[\psi]\leq A[\psi].
\end{equation}

To prove \eqref{eq:s8} set $f(t):=\cE_p(t\psi)$. Since $f$
has a minimum at $t=1$, we have $f'(1)=0$. However, the left side of
\eqref{eq:s8} is simply $f'(1)/2$.

To prove \eqref{eq:s9} introduce $g(t):= \cE_p(\psi_t)$ with
$\psi_t(x):=\psi(x/t)$. Again, $g$ has a minimum for $t=1$. Thus
\begin{equation}
  0= g'(1)= T[\psi]+3F[\psi]-2A[\psi]+5R[\psi]
\end{equation}
which proves \eqref{eq:s9}.

\subsection{Basics on existence, uniqueness, and excess charge in
  generalized TF theory}
In this section we are interested in the asymptotic behavior of the minimizing
density of Thomas-Fermi type functionals, however, with a more general
power $p$ and also a more general external potential 
\begin{equation}
  \label{eq:vg}
  V:= |\cdot|^{-1}*\mu
\end{equation}
with $\mu\in M$ where
\begin{definition}
  \label{eq:M}
  $M$ is the set of measures such that $\mathrm{supp}( \mu)$ compact, and
  $\mu(x)=\sum_{k=1}^KZ_k\delta(x-R_k)\rd x+\sigma(x)\rd x$ with
  $Z_1,...,Z_k\geq 0$, $ X_1,...,X_k\in\rz^3$, and
  $\ D[\sigma]<\infty$.
\end{definition}
 Note that the
molecular case as defined in \eqref{eq:V} is recovered by the choice
$\sigma=0$. The functional is defined as
\begin{equation}
  \label{eq:tfdomain}
  \gD_p^\mathrm{TF}:=\{\rho\in L^p(\rz^3)|\rho\geq0,\
  D[\rho]:=R[\sqrt\rho]<\infty\}.
\end{equation}

The generalized Thomas-Fermi functional is
\begin{equation}
  \label{eq:51}
  \begin{split}
    &\cE_{p}^\mathrm{TF}: \gD_p^\mathrm{TF} \to \rz\\
    &\rho \mapsto
     \tfrac1p\int_{\rz^3}\rho^{p} -\int_{\rz^3}V\rho + D[\rho].
  \end{split}
\end{equation}
Since in the atomic and molecular case the functional is unbounded
from below for $p\leq 3/2$, we are mainly interested in
$p>3/2$. Before we continue, note that $D(\cdot,\cdot)$, the
sesquilinear form associated with the quadratic form $D[\cdot]$, is a
scalar product on the set of all tempered distributions $\mu$ for
which $\int_{\rz^3}\rd \xi |\cF(\mu)(\xi)|^2/|\xi|^2<\infty$.  In the
following we will use $\eta:= \delta(|x|-R)/(4\pi R^2)$. We also write $p':=p/(p-1)$ for the dual power of $p$. Then
\begin{lemma}
  \label{lemma6}
  Assume $p>3/2$, $\mu\in M$, and $V=\mu*|\cdot|^{-1}$. Then
  $\cE_{p}^\mathrm{TF}$ is bounded from below and coercive in the
  $p$-norm and the Coulomb norm. More precisely
  \begin{equation}
    \cE_{p}^\mathrm{TF}(\rho)\geq \frac\gamma p \|\rho\|_p^p -{4\pi Z\over3-p'}\|\rho_p\| - 2 Z \sqrt{D[\rho]} + D[\rho].
  \end{equation}
\end{lemma}
\begin{proof}
  Write $V_>:= \eta*|\cdot|^{-1}$ and $V_<:=|\cdot|^{-1}-V_>$.  By
  H\"older's and Schwarz's inequality we have
  \begin{multline}
    \int_{\rz^3}\rd x V(x)\rho(x) = \sum_{k=1}^K Z_k\int_{\rz^3}\rd x \left(V_<(x-R_k)\rho(x)
      + V_>(x-R_k)\rho(x)\right) +2D(\rho,\sigma) \\
    \leq \sum_{k=1}^KZ_k\left( \|V_<\|_{p'}\|\rho\|_p + 2Z_k D(\eta,\rho)\right)
    +2D(\sigma,\rho)
  \end{multline}
  and thus by the Schwarz inequality
  \begin{equation}
    \cE_p^\mathrm{TF}(\rho) \geq\frac\gamma p \|\rho\|_p^p
    -Z\|V_<\|_{p'} \|\rho\|_p - Z\sqrt{D[\eta]D[\rho]} - \sqrt{D[\rho]D[\sigma]}+ D[\rho]
    \end{equation}
    from which the claim follows, since $0\leq V_<(x)\leq1/|x|$, and
    therefore in $L^{p'}$.
\end{proof}

The Euler equation for the minimizer $\rho^\mathrm{TF}$ of
\eqref{eq:51} -- generalizing the classical Thomas-Fermi equation --
reads
\begin{equation}
  \label{TFp}
  \rho^{p-1}= \phi_+,\ \phi=V-\rho*|\cdot|^{-1}.
\end{equation}

\begin{lemma}
  Assume $p>3/2$. Then $\cE_p^\mathrm{TF}$ has a unique
  minimizer $\rho^\mathrm{TF}$ in $\gD_p$. The minimizer fulfills
  \eqref{TFp}. Moreover, any solution $\rho^\mathrm{TF}\in \gD_p$ of
  \eqref{TFp} minimizes $\cE_p^\mathrm{TF}$.
\end{lemma}
As the proof is a standard compactness and convexity argument similar
to the classical Thomas-Fermi case (Lieb and Simon
\cite{LiebSimon1977}) we skip its proof and remark that it holds also
for the molecular case.

\begin{lemma}
   \label{lemma5}
   If $\rho^\mathrm{TF}$ is a minimizer of $\cE_p^\mathrm{TF}$ for
   $p>3/2$. Then $\int_{\rz^3}\rho^\mathrm{TF}=Z$ and $\varphi\geq0$. Moreover
   $\rho^\mathrm{TF}$ is spherically symmetric in the atomic case.
\end{lemma}
\begin{proof}
  The well known proof of the classical case using subharmonicity
  transcribes to the general case. Instead we give a simple
  variational proof for the atomic case: Suppose $Q<0$. Then, by
  \eqref{TFp}, there are $R,\epsilon>0$ such that
  $\varphi(x)\geq\epsilon/|x|$ for $|x|\geq R$. Thus
  $\rho^\mathrm{TF}(x)\gtrsim \rho_R(x):=|x|^{-{1\over p-1}}$ for $|x|>R$. But
  \begin{equation}
    D[\rho_R] \gtrsim  \int_R^\infty\rd r \int_R^\infty\rd s {r^2s^2\over
      \max\{r^{1\over p-1},s^{1\over p-1}\}}=\infty
  \end{equation}
  for $p\geq 4/3$.

  Suppose $Q>0$. Then, pick
  $R:=\inf\{s \big|\int_{|x|<s}\rd x\rho^\mathrm{TF}(x)\geq Z\}$. By assumption
  $R<\infty$. Now, define $\rho_R(x):=\rho^\mathrm{TF}(x)\theta(R-|x|)$ and
  $\rho_>:=\rho^\mathrm{TF}-\rho_R$ which by uniqueness of $\rho^\mathrm{TF}$ are spherically
  symmetric. We compute
  \begin{multline}
    \cE_p^\mathrm{TF}(\rho_R)-\cE_p^\mathrm{TF}(\rho^\mathrm{TF})
    < -\int_{\rz^3}\rd x {\rho_>(x)\over|x|} + 2D(\rho_R,\rho_>) -D[\rho_>]
    = -D[\rho_>]<0
  \end{multline}
  where the last equality holds because of Newton's theorem. Thus
  $\rho^\mathrm{TF}$ cannot be a minimizer which is a contradiction.

  By spherical symmetry and Newton's theorem it follows that $\varphi\geq0$.
\end{proof}
Note that we proved Lemma \ref{lemma5} by a simple and direct
variational argument using the spherical symmetry; however, the
theorem can be also be proven in the molecular case using subharmonic
estimates like Lieb and Simon \cite{LiebSimon1977} did. Since this
amounts to a mere transcriptions we skip it here.

\subsection{The Sommerfeld solution of generalized Thomas-Fermi theory
  and bounds on the physical solution}
In this appendix we allow for more general external potentials than
\eqref{eq:V}.

The classical atomic Sommerfeld solution
$S_\frac53(x)=9\gamma^3/(\pi^2|x|^4)$ solves
$ \Delta S_\frac53= 4\pi S_\frac53^\frac32$ on
$\dot{\rz}^3 :=\rz^3\setminus \{0\}$ and bounds the classical atomic
Thomas-Fermi potential from above. This generalizes to more general
$p$:
\begin{lemma}
  Pick $p\in(3/2,2)$
  \begin{equation}
    \label{eq:bsigma}
    b(p):=\left({(p-1)(3p-4)\over2\pi(2-p)^2}\right)^{p-1\over2-p},\ \text{and}\
    \sigma:={2(p-1)\over2-p}.
  \end{equation}
  Then
    \begin{equation}
    \label{Sommerfeld-generalize}
    S_p(x):= s_p(|x|):= b(p)|x|^{-\sigma}
  \end{equation}
  solves the generalized differential Thomas-Fermi equation
  \begin{equation}
    \label{gdTF}
    \Delta S_p= 4\pi S_p^{1\over p-1},
  \end{equation}
  on $\dot{\rz}^3$ and for the atomic case, i.e., $\mu=Z\delta$,
  \begin{equation}
    \label{eq:tfs}
    \phi\leq S_p
  \end{equation}
on $\dot\rz^3$.
\end{lemma}
\begin{proof}
  Inserting \eqref{Sommerfeld-generalize} into \eqref{gdTF} yields
  \begin{equation}
    {b(p)\over r^{2(p-1)+4-2p\over2-p}}
    \left[{2(p-1)\over2-p}{2p-2+2-p\over2-p}-2{2(p-1)\over2-p}
    \right]
    =4\pi b(p)^{1\over p-1}r^{-{2\over2-p}}
  \end{equation}
  which is equivalent with
  \begin{equation}
     b(p)
    \left[{2(p-1)\over2-p}{p\over2-p}-2{2(p-1)\over2-p}
    \right]
    =4\pi b(p)^{1\over p-1}
  \end{equation}
  or
   \begin{equation}
     {2(p-1)(3p-4)\over(2-p)^2} =4\pi b(p)^{2-p\over p-1}.
   \end{equation}
   Inserting the definition of $b(p)$ from
   \eqref{Sommerfeld-generalize} gives identical left and right
   sides of \eqref{gdTF}.

   The bound \eqref{eq:tfs} follows then from the maximum principle,
   since the inequality is certainly true in a neighborhood of zero,
   since $\phi(x)\leq Z/|x|$ and both $\phi(x),S_p(x)\to0$ as
   $x\to\infty$.
 \end{proof}

 \begin{lemma}
   \label{simple-bound}
   Suppose $p\in(3/2,2)$, $V$ as in Lemma \ref{lemma6},
   $\rho^\mathrm{TF}$ is the minimizer of $\cE_p^\mathrm{TF}$,
   $\phi:=V-\rho^\mathrm{TF}*|\cdot|^{-1}$, and
  \begin{equation}
    \label{eq:ap}
    a(p):=\left({(p-1)p\over2\pi(2-p)^2}\right)^{p-1\over2-p},\ \text{and}\
    \sigma:={2(p-1)\over2-p}.
  \end{equation}
  Then for $|x|>R$,
  \begin{equation}
    \label{differential-TF-equation}
    \Delta\phi=4\pi\phi^\frac1{p-1},
  \end{equation}
  and
  \begin{equation}
    \label{sommerfeldR}
    \phi(x) \leq S_{p,R}(x):=s_{p,R}(|x|):={a(p)\over (|x|-R)^\sigma}.
  \end{equation}
\end{lemma}
\begin{proof}
  The generalized Thomas-Fermi equation \eqref{TFp} and Poisson's
  equation $\Delta\varphi=4\pi\rho^\mathrm{TF}$ imply
  \eqref{differential-TF-equation}.

  To show \eqref{sommerfeldR} we first show that  on
  the complement of $\overline{B_R(0)}$
  \begin{equation}
    \Delta S_{p,R}
    \leq 4\pi S_{p,R}^{1\over p-1}:
  \end{equation}
  \begin{equation}
    \begin{split}
   & s_{p,R}''(r) +\tfrac2rs_{p,R}'(r)= \sigma(\sigma+1){a(p)\over (r-R)^{\sigma+2}}
    -2\sigma{a(p)\over r(r-R)^{\sigma+1}}\\
    \leq  &{\sigma(\sigma+1)\over4\pi}
      a(p)^{-{2-p\over p-1}} 4\pi s_{p,R}(r)^{1\over p-1}=4\pi s_{p,R}(r)^{1\over p-1}
    \end{split}
  \end{equation}
  and that the inequality is true on $\partial B_R(0)$ and both
  $S_{p,R} $ and $\phi$ tend to zero at infinity. Thus,
  the inequality follows for all $|x|>R$ by subharmonicity.
\end{proof}

Now we turn the Sommerfeld solution with a leading remainder term
extending a result by Brezis and Lieb \cite{BrezisLieb1979}, Solovej
\cite[Lemma 11]{Solovej1990}\cite[Lemma 4.4]{Solovej2003}. We will
largely follow his proof.
\begin{lemma}
    Pick $p\in(3/2,2)$, $V$ as in Lemma \ref{lemma6},
  \begin{equation}
    \label{zeta}
  \zeta:={-5p+6 + \sqrt{p^2+20p-28}\over2(2-p)},
\end{equation} and a
  smooth function $\pi$ on $B_R(0)^c$ fulfilling
  \begin{equation}
    \Delta\phi = 4\pi\phi^{1\over p-1}.
  \end{equation}
  Moreover, define
  \begin{align}
    a(R) := & \liminf_{r\searrow}\sup_{|x|=r}\left[\left({\phi(x)\over s_{p,\gamma}(|x|)}\right)^{-\frac12}-1\right]|x|^\zeta,\\
    A(R):=&\liminf_{r\searrow}\sup_{|x|=r}\left[{\phi(x)\over s_{p,\gamma}(|x|)}-1\right]|x|^\zeta
  \end{align}
  Then on $B_R(0)^c$
  \begin{equation}
    \label{Sommerfeld2.0}
    (1+a(R)|x|^{-\zeta})^{-{p-1\over p-2}} s_{p,\gamma}(|x|)
    \leq\phi(x)\leq(1+A(R)|x|^{-\zeta}) s_{p,\gamma}(|x|).
    \end{equation}
  \end{lemma}
  Note that for the classical exponent $p=5/3$, $\zeta=(\sqrt{73}-7)/2$ which agrees with \cite[Lemma 4.4]{Solovej2003}.
\begin{proof}
  We start by proving that $\phi$ tends to zero at infinity. To
  this end pick $L> 4R$ and define the function
  $f(x):= C (s_{p,\gamma}(r-L/4)+s_{p,\gamma}(L-r))$ on $(L/4,L)$ writing $r:=|x|$. 
  \begin{multline}
    \Delta f(x) = C ( S''(r-L/4)+\frac2r S'(r-L/4)+  S''(L-r)-\frac2r S'(L-r))\\
   = Cb(p)\gamma^\frac1{2-p}\left[ {2(p-1)(2p-2+2-p)\over (2-p)^2}{1\over(r-\tfrac L4)^\frac{2p-2+4-2p}{2-p}}
      -\frac2r {2(p-1)\over 2-p} {1\over (r-\tfrac L4)^\frac p{2-p}}
    \right.\\
    \left.  +{2(p-1)(2p-2+2-p)\over (2-p)^2}{1\over(L-r)^\frac{2p-2+4-2p}{2-p}}
      +\frac2r {2(p-1)\over 2-p} {1\over (L-r)^\frac p{2-p}}
    \right]\\
    = Cb(p)\gamma^\frac1{2-p}{2(p-1)\over 2-p}\\
    \times \left[ {p\over 2-p}{1\over(r-\tfrac L4)^\frac{2}{2-p}}
       -\frac2r {1\over (r-\tfrac L4)^\frac p{2-p}}
      +{p\over 2-p}{1\over(L-r)^\frac2{2-p}}
       +\frac2r  {1\over (L-r)^\frac p{2-p}}
     \right].
   \end{multline}
   We wish to estimate the sum of the second and fourth term of the
   last bracket by a multiple of the third one, i.e., show that
   \begin{equation}
     -\frac2r {1\over (r-\tfrac L4)^\frac p{2-p}} +\frac2r  {1\over (L-r)^\frac p{2-p}}\leq D {1\over(L-r)^\frac2{2-p}}.
   \end{equation}
   By scaling it is sufficient to show this for $L=1$, i.e.,
 \begin{equation}
     -\frac2r {1\over (r-\tfrac14)^\frac p{2-p}} +\frac2r  {1\over (1-r)^\frac p{2-p}}\leq D {1\over(1-r)^\frac2{2-p}}
   \end{equation}
   which is equivalent with
   \begin{equation}
    - 2(1-r)^\frac2{2-p}+2(1-r)(r-\tfrac14)^\frac p{2-p}\leq D r (r-\tfrac14)^\frac p{2-p}
   \end{equation}
   or
    \begin{equation}
     -2(1-r)^\frac2{2-p}+2(r-\tfrac14)^\frac p{2-p}\leq (D-2) r (r-\tfrac14)^\frac p{2-p}
   \end{equation}
   which is fulfilled if
    \begin{equation}
     -2(1-r)^\frac2{2-p}+2(r-\tfrac14)^\frac p{2-p}\leq {D-2\over4}(r-\tfrac14)^\frac p{2-p}
   \end{equation}
   holds, since $r>L/4$. This is equivalent with
    \begin{equation}
     -2(1-r)^\frac2{2-p}\leq \left(\frac D4-\frac52\right)(r-\tfrac14)^\frac p{2-p}=0
   \end{equation}
   which is true, since we picked $D=10$.
   Thus we get altogether
   \begin{multline}
     \Delta f(x)\leq Cb(p)\gamma^\frac1{2-p}{2(p-1)\over 2-p} \left[
       {p\over 2-p}{1\over(r-\tfrac L4)^\frac{2}{2-p}} +{20-10p+p\over
         2-p}{1\over(L-r)^\frac2{2-p}} \right]\\
    = Cb(p)\gamma^\frac1{2-p}{2(p-1)\over (2-p)^2} \left[
       p{1\over(r-\tfrac L4)^\frac{2}{2-p}} +(20-9p){1\over(L-r)^\frac2{2-p}} \right].
   \end{multline}
Moreover, for $p\in(1,2)$
\begin{equation}
  f(r)^\frac1{p-1}\geq C^\frac1{p-1}\left(s_{p,\gamma}(r-L/4)^\frac1{p-1} \right).
\end{equation}
Thus, for each $p$ and $\gamma$ there is a constant $C$ such that for
all $L>4R$ $f$ is a supersolution of the generalized differential TF
equation, i.e.,
\begin{equation}
  \Delta f \leq 4\pi f^{1\over p-1}.
\end{equation}
on $M_L:=\overline{B_L(0)}\setminus B_{L/4}(0)$.  Thus the maximum
principle implies $f\geq \phi$ on $A_L$ and
$\sup_{|x|=L/2}\phi(x)\lesssim_{p} L^{2(p-1)\over2-p}$. Thus
$\phi(x)|x|^{2(p-1)\over2-p}$ is bounded for $|x|>R$.

Next we turn to the main estimate. For any $k\in \rz$ we define
\begin{align}
  \omega^+_k(x):=&(1+k|x|^{-\zeta})s_{p}(|x|), \label{+}\\
  \omega^-_{k}(x):= &(1+k|x|^{-\zeta})^{-{p-1\over p-2}}s_{p}(|x|).\label{-}
\end{align}
If we pick $R'>R$ and set $k=A':=A(R')$ and $k=a':=a(R')$ then
$\omega^+_{A'}$ and $\omega^+_{a'}$ are right and left side of the
main estimate \eqref{Sommerfeld2.0} but at $R'$ instead of $R$.

We claim that
\begin{equation}
  \label{omega}
  \Delta \omega^+_{k}\leq 4\pi (\omega^+_{k}/\gamma)^{1\over p-1}\ \text{and}\
  \Delta \omega^-_{k}\geq 4\pi (\omega^-_{k}/\gamma)^{1\over p-1}.
\end{equation}
(The first inequality for $p=5/3$ was known to Brezis and Lieb (first
inequality above Proposition A.5 in \cite{BrezisLieb1979}).)

Since $\omega^+_{A'}(R')=\phi(R')=\omega^-_{a'}(R')$ and since both
functions tend to zero at infinity, this would imply the claim is true
by a standard application of the maximum principle and taking the
limit $R'\to R$.

Thus it remains to show that \eqref{omega} is really true. We compute starting with the supersolution:
\begin{multline}
  \Delta \omega^+_k(x)\\
  = \left(S''(|x|)+{2\over|x|}S'(|x|)\right)(1+{k\over |x|^\zeta})
  -2S'(|x|) k {\zeta\over|x|^{\zeta+1}} +s_{p,\gamma}(|x|)k{\zeta(\zeta+1)-2\zeta\over|x|^{\zeta+2}}\\
  =4\pi\left(s_{p,\gamma}(r)\over\gamma\right)^{1\over p-1}(1+{k\over r^\zeta}) + s_{p,\gamma}(r)k\zeta r^{-\zeta-2}\left( 2^2{p-1\over p-2}+ \zeta-1\right)\\
  =4\pi\left(s_{p,\gamma}(r)\over\gamma\right)^{1\over p-1}(1+{k\over r^\zeta}) + s_{p,\gamma}(r)k\zeta \left({s_{p,\gamma}(r)\over b(p) \gamma^{1\over2-p}}\right)^{2-p\over p-1}r^{-\zeta}\left( {5p-6\over p-2}+ \zeta\right)\\
  =4\pi\left(s_{p,\gamma}(r)\over\gamma\right)^{1\over p-1}\left[1+{k\over r^\zeta} +
    {k\zeta\over 4\pi b(p)^{2-p\over p-1}r^\zeta}\left({5p-6\over p-2}+ \zeta\right)\right]\\
  =4\pi\left(\omega^+_k\over\gamma\right)^{1\over p-1}
  {1+{k\over r^\zeta}\left[ 1+{2\pi(p-2)^2\zeta\over 4\pi (p-1)(3p-4)}\left({5p-6\over p-2}+ \zeta\right)\right]\over (1+kr^{-\zeta})^{1\over p-1}}.
\end{multline}
Thus, the claim would follow if the big fraction of the last line is not bigger than one. Since $p-1\leq1$, it suffices to show that
\begin{equation}
  1+{k\over r^\zeta}\left[ 1+{(p-2)^2\over 2 (p-1)(3p-4)}\zeta\left({3p-2\over p-2}+ \zeta\right)\right]\leq 1 +(p-1)^{-1}k r^{-\zeta}
\end{equation}
or
\begin{equation}
{(p-2)^2\over 2(3p-4)}\zeta\left({5p-6\over p-2}+ \zeta\right)\leq 2-p
\end{equation}
which is fulfilled, if
\begin{equation}
  \zeta^2+{5p-6\over 2-p}\zeta = 2{3p-4\over 2-p}
\end{equation}
which true for the chosen $\zeta$.

Next we treat the subsolution
\begin{multline}
  \Delta \omega^-_{a'}(x)\\
  \geq \left(S''(r)+{2\over|x|}S'(r)\right)\left(1+{a'\over
      r^\zeta}\right)^{-{p-1\over2-p}}
  +2S'(r) {p-1\over2-p}\left(1+{a'\over r^\zeta}\right)^{-{1\over2-p}}{a'\zeta\over r^{\zeta+1}}\\
  +s_{p,\gamma}(r){p-1\over2-p}\left[-\left(1+{a'\over r^{\zeta}}\right)^{-{1\over
      2-p}}{a'\zeta(\zeta+1)\over r^{\zeta+2}}
  +2\left(1+{a'\over r^\zeta}\right)^{-{1\over2-p}}{a'\zeta\over r^{\zeta+2}}\right]\\
 ={4\pi\over\gamma^{1\over p-1}}s_{p,\gamma}(r)^{1\over p-1}\left(1+{a'\over
     r^\zeta}\right)^{-{p-1\over2-p}}\\
+{s_{p,\gamma}(r)\over r^2}{(p-1)\over2-p}\left[-{4(p-1)\over2-p}\left(1+{a'\over r^\zeta}\right)^{-{1\over2-p}}{a'\zeta\over r^\zeta} -\left(1+{a'\over
      r^{\zeta}}\right)^{-{1\over 2-p}}{a'\zeta(\zeta-1)\over r^{\zeta}}\right]\\
={4\pi\over\gamma^{1\over p-1}}s_{p,\gamma}(r)^{1\over
  p-1} \left\{\left(1+{a'\over r^\zeta}\right)^{-{p-1\over2-p}}
  +{2\pi(2-p)^2\over 4\pi(p-1)(3p-4)}{(p-1)\over2-p}\right.\\
\times \left[-{4(p-1)\over2-p}\left(1+{a'\over
      r^\zeta}\right)^{-{1\over2-p}}{a'\zeta\over r^\zeta}
  \left.-\left(1+{a'\over
          r^{\zeta}}\right)^{-{1\over 2-p}}{a'\zeta(\zeta-1)\over r^{\zeta}}\right]\right\}\\
 ={4\pi\over\gamma^{1\over p-1}}s_{p,\gamma}(r)^{1\over
    p-1}\left\{\left(1+{a'\over
        r^\zeta}\right)^{-{p-1\over2-p}}\right.\\
\left.+\frac12{2-p\over3p-4} \left[-{4(p-1)\over2-p}\left(1+{a'\over
      r^\zeta}\right)^{-{1\over2-p}}{a'\zeta\over r^\zeta}
  -2\left(1+{a'\over
          r^{\zeta}}\right)^{-{1\over 2-p}}{a'\zeta(\zeta-1)\over r^{\zeta}}\right]\right\}\\
  ={4\pi\over\gamma^{1\over p-1}}\omega^-_{a'}(r)^{1\over
    p-1}\left\{1+\left[1 -2{p-1\over3p-4}\zeta
      -\frac12{2-p\over3p-4}\zeta(\zeta-1)\right]{a'\over r^{\zeta}}\right\}
  ={4\pi\over\gamma^{1\over p-1}}\omega^-_{a'}(r)^{1\over
    p-1}
\end{multline}
where we drop a non-negative summand containing $(a'\zeta)^2$ in the
first inequality and use \eqref{zeta} in the last step.  Thus
$\omega^-_{a'}$ is indeed a subsolution.  The result follows by taking
the limits $\lim_{r'\searrow r}a(r')$ and $\lim_{r'\searrow r}A(r')$.
\end{proof}

Equipped with two atomic supersolutions, namely $S_{p,R}$ and
$\omega^+_k$, of the generalized Thomas-Fermi equation suggests that
the pointwise minimum bounds $\phi$ from above. This would improve the
bound for large $r$, since the coefficient $a(p)$ of the leading term
of $S_{p,R}$ for large $r$ is larger than the coefficient $b(p)$ of
the leading term of $\omega_k^+$. In fact this is true:
\begin{lemma}
  \label{better-bound}
  Assume $p\in(3/2,2)$,$r>R>0$, $k\in\rz^3$, $\zeta$ as in
  \eqref{zeta}, and $\omega^+_k$ as defined in \eqref{+}. Then
  \begin{equation}
    \label{sbound}
    \phi(r)\leq \sigma_p(r):=\min\{s_{p,R}(r),\omega^+_k(r)\}.
  \end{equation}
\end{lemma}
\begin{proof}
  Both functions are supersolutions of the Thomas-Fermi equation and $s_{p,R}(r)\geq \phi(r)$. Moreover, the two function have exactly one point of intersection $r_0$ for $r>R$.
  Thus $\omega_k^+(r_0)=s_{p,R}(r_0)\geq\phi(r_0)$. Since $\omega_k^+(r)\to0$ as $r\to\infty$ we have $\omega_k^+(r)\geq\phi(r)$ for $r\geq r_0$. 
\end{proof}

%\bibliographystyle{plain}
%\bibliography{coulomb}
	
\def\cprime{$'$}

\end{document}